\documentclass[%
 reprint,
 amsmath,amssymb,
 aps,
pra,
]{revtex4-2}

\usepackage{graphicx}% Include figure files
\usepackage{dcolumn}% Align table columns on decimal point
\usepackage{bm}% bold math
\usepackage{braket}
\usepackage{amsmath,amssymb,amsfonts,bbm}%
\usepackage{amsthm}%
\usepackage{mathrsfs}%
\usepackage{xcolor}%
\usepackage{algorithm}%
\usepackage{algorithmicx}%
\usepackage{algpseudocode}%
\usepackage{hyperref}
\usepackage[nameinlink,poorman]{cleveref}
\usepackage{todonotes}

\theoremstyle{thmstyleone}%
\newtheorem{theorem}{Theorem}%
\newtheorem{lemma}[theorem]{Lemma}% 
\theoremstyle{thmstyletwo}%
\theoremstyle{thmstylethree}%
\newtheorem{definition}{Definition}%
\crefname{equation}{Eq.}{Eqs.}
\crefname{figure}{Fig.}{Figs.}
\crefname{table}{Table}{Tables} 
\crefname{section}{Section}{Sections}
\crefname{chapter}{Chapter}{Chapters}
\crefname{appendix}{Supplementary Note}{Supplementary Notes}
\crefname{algorithm}{Algorithm}{Algorithms}
\crefname{theorem}{Theorem}{Theorems}
\crefname{defn}{Definition}{Definitions}
\crefname{definition}{Definition}{Definitions}
\crefname{azm}{Assumption}{Assumptions}
\crefname{corollary}{Corollary}{Corollaries}
\crefname{lemma}{Lemma}{Lemmas}
\crefname{thmprop}{property}{properties}
\crefname{proposition}{Proposition}{Propositions}
\crefname{remark}{Remark}{Remarks}

\begin{document}

\preprint{APS/123-QED}

\title{Characterizing noisy quantum computation with imperfectly addressed errors}

\author{Riddhi S. Gupta}
\email{riddhi.gupta@uq.edu.au}
\affiliation{School of Mathematics and Physics, The University of Queensland Brisbane QLD 4072 Australia \relax}
\author{Salini Karuvade}
\affiliation{School of Physics, The University of Sydney, NSW 2006, Australia}
\author{Kerstin Beer}
\affiliation{School of Mathematical and Physical Sciences, Macquarie University, Sydney, New South Wales 2109, Australia}
\author{Laura J. Henderson}
\affiliation{School of Mathematics and Physics, The University of Queensland Brisbane QLD 4072 Australia}
\author{Sally Shrapnel}
\affiliation{School of Mathematics and Physics, The University of Queensland Brisbane QLD 4072 Australia}

\date{\today}

\begin{abstract}
 Quantum protocols on hardware are subject to noise that prohibits performance. Protocols for addressing errors, such as error correction or error mitigation, may fail to combat errors in quantum computation if noise violates critical assumptions required for these protocols to be effective. However, tools for characterizing such failures in realistic operating conditions are limited. For example, while brute force simulations may be used to characterize the impact of such failures on a handful of input states, such simulations lack a complete description for how noise transforms state-spaces in the full quantum Hilbert space. In this work, we associate quantum computation subject to realistic noise to an ensemble of random superoperators and study the eigen- and singular spectral distributions over this ensemble. We propose a new theoretical framework to characterize singular values of random complex matrices using matrix Chernoff concentration. Using our framework, we analyze imperfectly addressed errors in error mitigation and error correction. We find that distributions of singular spectra depend on how noise violates critical assumptions of these protocols. Finally, we quantitatively discuss how our work may be applied to understanding limiting behavior of quantum computation, such as establishing spectral gaps and relaxation times for specific families of quantum Markov processes. Our work paves the way for new tools to diagnose when to trust the output of noisy quantum computers.
\end{abstract}

\maketitle

\section{\label{sec:intro} Introduction}

While quantum computing experiments continue to demonstrate remarkable progress in underlying hardware capabilities, combating the effects of noise on computation is critical to trusting the output of quantum hardware. Two dominant strategies for addressing noise arise as quantum error correction \cite{Gottesman1997May} or quantum error mitigation \cite{Temme2017Nov,vandenBerg2023Aug}. While error correction detects and corrects local errors in quantum states under operating regimes with sufficiently low-error, error mitigation yields unbiased estimators of noise-free expectation values while tolerating higher noise levels. Indeed both error correction and mitigation protocols assume that intrinsic noise has certain properties that may not hold true under realistic operating conditions. Examples include, e.g., that intrinsic noise is Markov; that intrinsic noise is sufficiently weak; or that intrinsic noise has restricted impact or is restricted to certain types of operations. Indeed quantum protocols for addressing noise may themselves be subject to pathological conditions which violate critical assumptions required for protocols to be effective. Tools for simulating quantum computations under these kinds of realistic operating conditions are limited.

One typical approach for simulating noisy quantum computation is to use channel formalism represented by superoperators, i.e. complex matrices, acting on vectorized quantum states \cite{watrous_theory_2018}. Previous literature has established distance metrics to compare ideal versus noisy quantum channels \cite{Gilchrist2005Jun} using induced Schatten norms that depend on singular values or the numerical ranges of operators \cite{Watrous2004Nov,Johnston2007Nov,Rastegin2009May,Chelstowski2023Aug}. Aside from such distance metrics, it is also well recognized that channel superoperators may define transition matrices for Markov chains on non-commutative probability spaces \cite{temme_2-divergence_2010,wolf_tour_2012}. Existing work on quantum Markov processes establishes the limiting behavior of specific types of quantum channels under repeated applications, including how quickly initial information is forgotten by a channel \cite{terhal_problem_2000,burgarth_ergodic_2013,albert_asymptotics_2019}. Indeed the spectral properties of quantum channels are important for characterizing e.g. subspaces in error correction \cite{ippoliti_perturbative_2015}, mixing times, relaxation times and fixed points of mixing channels \cite{temme_2-divergence_2010,terhal_problem_2000,burgarth_ergodic_2013,albert_asymptotics_2019}; spectral gaps of random channels \cite{bruzda_random_2009,kukulski_generating_2021}, random Lindblad dynamics \cite{Can2019Nov}, and random isometries \cite{fischmann_induced_2012,gonzalez-guillen_spectral_2018,Karuvade2018Mar}. Despite the importance of eigen- or singular values to establishing both channel norms and limiting behavior of channels, they have not been studied to understand when noisy quantum protocols fail. 

In this work, we motivate the study of eigen- and singular values of noisy quantum computation subject to realistic, imperfect operating conditions. We consider an ensemble of random superoperators where each sample represents an experimental run probabilistically subject to noise that may break a protocol's underlying assumptions. For example, for error correction, the number of single qubit errors may exceed code distance, while for error mitigation, a learned model may drift during mitigation. We call such random processes failure modes. As a first step towards operational usefulness, we propose a general theoretical framework to provide an upper bound to any ordered singular value of ensembles of random channel superoperators. Our analysis holds even if random superoperators only satisfy Kraus conditions in an idealized limit where the probability of incurring a failure mode goes to zero. To illustrate potential applications of our framework, we empirically find that both the singular spectral radius of imperfectly error-mitigated quantum gates as well as the average singular multiplicities of imperfectly error-corrected quantum codes are correlated with different failure modes for these protocols. We discuss how these insights could be useful for comparing operational robustness of different quantum protocols and/or establish their limiting behavior. 

The structure of this document is as follows. \cref{sec:background} provides the groundwork for how classical eigenspectral concentration results for random Hermitian matrices are applied to random complex matrices. In \cref{sec:applications}, we present our lemma to upper bind the maximal singular value of general random matrices with concentration. We apply this lemma to superoperators associated with quantum channels for error correction and error mitigation. We characterize the affect of failure modes empirically and compare with theoretical results. In \cref{sec:fullversion} we extend this lemma to a theorem that provides an upper bound for any singular value in an ordered set with concentration. We discuss potential for future operational applications before concluding in \cref{sec:conclusion}.

\section{\label{sec:background} Background}

To understand the spectrum of general quantum channels, we will need to elevate results applicable to Hermitian matrices to the general class of complex matrices using a Hermitian dilation \cite{tropp_introduction_2015}. In terms of notation, we denote the eigenvalues of a general complex square matrix $\hat{\Lambda} \in \mathbb{C}^{d' \times d'}$ as the set of all complex numbers that satisfy $\mathrm{det}(\lambda\mathbb{I}_{d'} - \hat{\Lambda})=0$. We use $\lambda(\hat{\Lambda} )$ to denote the full set of eigenvalues of the operator $\hat{\Lambda} $, as well as to denote the $i$-th ordered eigenvalue $\lambda_i(\hat{\Lambda} )$. The set of Hermitian operators are denoted by $\chi \in \mathbb{C}^{d' \times d'}$ and the set of non-negative Hermitian operations are denoted by $X \in \mathbb{C}^{d' \times d'}$. 

\begin{definition} A Hermitian dilation $\mathcal{M}_{d'} \to \mathcal{M}_{2d'}$ of a complex $d'\times d'$ matrix $\hat{\Lambda}$ is
\begin{align}
   \chi := \begin{bmatrix}
        0 & \hat{\Lambda} \\
        \hat{\Lambda}^\dagger & 0
    \end{bmatrix} \label{th:dilation},
\end{align} where $0$ represents sub-blocks, $\hat{\Lambda}^\dagger$ is a conjugate transpose of $\hat{\Lambda}$, and $\chi$ is Hermitian.
\end{definition}

Note that the action of the dilation is to relate the results for eigenvalues of Hermitian matrices to singular values of complex matrices. Singular values $\sigma_i$ of a complex matrix, $\hat{\Lambda}$ are the square-root of eigenvalues of the products,
\begin{align}
    \sigma_i(\hat{\Lambda}) := \sqrt{\lambda_i(\hat{\Lambda}^\dagger \hat{\Lambda})} = \sqrt{\lambda_i(\hat{\Lambda} \hat{\Lambda}^\dagger)}, \quad \forall i.
\end{align} The right (left) singular vectors are the eigenvectors of $\hat{\Lambda}^\dagger \hat{\Lambda}$ ($\hat{\Lambda} \hat{\Lambda}^\dagger$). We use $\sigma(\hat{\Lambda})$ to denote the full set of singular values of an operator $\hat{\Lambda}$, and the $i$-th ordered singular in this set as $\sigma_i(\hat{\Lambda})$. When $\hat{\Lambda}$ is normal ($[\hat{\Lambda}, \hat{\Lambda}^\dagger]=0$), then singular values reduce to the absolute eigenvalues of $\hat{\Lambda}$. Normal operators include Hermitian, unitary, and projection operators.

We will now make straightforward use of Chernoff matrix concentration results \cite{tropp_introduction_2015}. Of these, the main result we leverage in this work is the concentration of extremal spectral values of a sequence of random matrices, restated from existing literature below. 

\begin{theorem} [Matrix Chernoff \cite{tropp_introduction_2015}]\label{th:matrixchernoff}
  Consider a finite sequence $X_k$ of independent, random Hermitian matrices with a common dimension $d'$. Assume that eigenvalues $0 \leq \lambda_{\mathrm{min}}(X_k)$ and $ \lambda_{\mathrm{max}}(X_k) \leq L$ for each $k$. Introduce a matrix $Y = \Sigma_k X_k$ and define the maximum eigenvalue of the expectation $\mathbb{E}Y$ as,
    \begin{align}
        % \mu_{\mathrm{min}} &= \lambda_{\mathrm{min}}(\mathbb{E}Y), \\
        \mu_{\mathrm{max}} &= \lambda_{\mathrm{max}}(\mathbb{E}Y).
    \end{align} Then for $\theta > 0$, and $\epsilon \geq 0$
    \begin{align}
        % \mathbb{E} \lambda_{\mathrm{min}}(Y) \geq \frac{1-e^{-\theta}}{\theta} \mu_{\mathrm{min}}- \frac{L}{\theta}\log(D), \\
        &\mathbb{E} \lambda_{\mathrm{max}}(Y) \leq \frac{e^\theta -1}{\theta} \mu_{\mathrm{max}}+ \frac{L}{\theta}\log(d'), \label{eqn:mc:maxeigenvalue} \\ 
% \end{align}
% Moreover, 
% {\small \begin{align}
% \mathbb{P} \left[ \lambda_{\mathrm{min}}(Y)) \leq (1-\epsilon)\mu_{\mathrm{min}}\right] & \leq D \left(\frac{e^{-\epsilon}}{(1- \epsilon)^{1- \epsilon}}\right)^{\mu_{\mathrm{min}}/L}, \epsilon \in [0,1) \\
& \mathbb{P} \left[ \lambda_{\mathrm{max}}(Y)) \geq (1+\epsilon)\mu_{\mathrm{max}}\right] \leq d' \left( \frac{e^{\epsilon}}{(1 + \epsilon)^{1+ \epsilon}}\right)^{\mu_{\mathrm{max}}/L}, \label{eqn:mc:maxeigenvalueprob}
    \end{align}
    % \proof See Chapter 5 of Ref. ~\cite{tropp_introduction_2015}.
\end{theorem} The above theorem, the matrix Chernoff inequality, gives concentration bounds for the largest eigenvalue of a sum of independent, random, positive semi-definite matrices. (It is a matrix analogue of the classical Chernoff bound, which controls the tail probabilities for sums of bounded independent scalar random variables). The Chernoff fluctuation term $\mu_{\mathrm{max}}$ represents the fluctuations in the spectrum of the average random matrix $\lambda_{\mathrm{max}}(\mathbb{E}Y)$, to be contrasted with the expected value of the extremal values of the spectrum $\mathbb{E} \lambda_{\mathrm{max}}(Y)$ in \cref{eqn:mc:maxeigenvalue}. Here $L$ is a uniform upper bound, while $d'$ is related to matrix dimension under consideration. The choice of $\theta > 0$ provides some control over the sharpness of the upper bound in \cref{eqn:mc:maxeigenvalue} for specific applications. Meanwhile, concentration of the maximal eigenvalue is given by \cref{eqn:mc:maxeigenvalueprob}, where the upper tail decays faster than an exponential random variable with a mean $L/\mu_{\mathrm{max}}$. Corresponding results for the minimal eigenvalue can be found in Ref.~\cite{tropp_introduction_2015} but are not used in this work.

Additionally, it will be helpful to partition spectral values, e.g. to isolate the second-largest absolute spectral value. To do this, we use the general form of the Cauchy interlace theorem.
\begin{theorem}[Generalization of Cauchy Interlace Theorem\cite{hwang2004cauchy,higham2021}] \label{th:cauchyinterlace}
    For Hermitian matrix $\chi \in \mathbb{C}^{d' \times d'}$, let $\chi_S$ be a principal submatrix of order $d'-l$. Then $\lambda_{i+l}(\chi) \leq \lambda_i(\chi_S) \leq \lambda_i(\chi)$ for $i = 1, \hdots, d'-l$, counting with multiplicity.
 \end{theorem} This theorem states that the eigenvalues of the principal submatrix $\chi_S$, formed by deleting $d'-l$ rows (and the same $d'-l$ columns) of a Hermitian matrix $\chi$, will interlace with the eigenvalues of $\chi$. 

In the next section, we present a tool to study spectral fluctuations of general complex matrices, $\hat{\Lambda}$, which will later be associated with superoperators of quantum channels. A quantum channel is a completely positive map represented by the following operator sum, 
    \begin{align}
    \Lambda (\rho) := \Sigma_{i=1}^\kappa A_i \rho A_i^\dagger, \quad \Sigma_i A_i^\dagger A_i \leq \mathbb{I}, \label{eqn:krausrep}
\end{align} where $\rho$ is an $d\times d$ dimensional input quantum state, and $A_i \in \mathcal{M}_d $ are Kraus operators for a total number of $\kappa$ terms in the sum. For $n$ qubits, $d=2^n$ is considered in this work. The channel is trace preserving $ \iff \Sigma_i A_i^\dagger A_i = \mathbb{I}$ or unital $ \iff \Sigma_i A_i A_i^\dagger = \mathbb{I}$, where $\mathbb{I}$ is the identity. The channel $\Lambda$ has a matrix representation $\hat{\Lambda} \in \mathbb{C}^{d^2 \times d^2}$, with matrix elements given by $\hat{\Lambda}_{ij}:= \langle B_i |\Lambda |B_j\rangle$ for an orthonormal basis $\{B_i\}$ of $\mathcal{M}_{d^2}$, i.e. the Hilbert-Schmidt norm in regular operator space. This basis will refer to matrix or natural units in this work \cite{wolf_tour_2012,watrous_theory_2018}. One obtains a superoperator representation of the quantum channel from \cref{eqn:krausrep} as the complex matrix,
    \begin{align}
        \hat{\Lambda} := \Sigma_{i=1}^\kappa  A_i^* \otimes A_i,  \label{eqn:superoprep}
    \end{align} where ${}^*$ denotes entry-wise conjugation. In particular, superoperators of quantum channels need not be normal \cite{wolf_tour_2012}. For completely-positive trace-preserving maps, the maximum singular value can be between unity and the dimension of the Hilbert space, and is exactly unity for unital channels \cite{watrous_theory_2018}. Finally, the composition of two channels $\Lambda_1 \circ \Lambda_2 (\rho)$ can be expressed as matrix multiplication of superoperators, $\hat{\Lambda}_1 \hat{\Lambda}_2$. Further details on channel formalism are provided in \cref{app:channels}. We return to superoperators for quantum channels after first considering singular values of general complex matrices. 

\section{Characterizing imperfect protocols for combating noise \label{sec:applications}}
 
We now present a tool to study spectral fluctuations of general complex random matrices, $\hat{\Lambda}$. Presented as \cref{th:lemma}, this result is a straightforward combination of matrix Chernoff results applied to a Hermitian dilation of a complex matrix. Our result will be subsequently used to study quantum channels subject to noise including applied failure modes.

\begin{lemma}[Spectral concentration of maximal singular values]\label{th:lemma}
    Let $\hat{\Lambda} \in \mathcal{M}_{d'}$ be a random complex valued matrix with ordered singular values $ \sigma_1(\hat{\Lambda}) \geq \hdots \sigma_i(\hat{\Lambda}) \hdots  \geq \sigma_{d'}(\hat{\Lambda})$, counted with multiplicity. Introduce constants $\theta, L \geq 0$, as in \cref{th:matrixchernoff}. Then the maximal singular value, $\sigma_1(\hat{\Lambda})$ is upper bounded,
    \begin{align}
    \mathbb{E}[\sigma_{1}(\hat{\Lambda})] &\leq \mu \left( \frac{ e^\theta -1 }{\theta}\right) + \frac{2L}{\theta}\log (2d'),
    \end{align} 
    where $\mu$ is defined by \cref{eqn:lemma:chernoffmu}, and the expectation value is computed over an ensemble of random channels.
\proof One considers the Hermitian dilation \cref{th:dilation} with $\chi : = \begin{bmatrix}
    0 & \hat{\Lambda} \\
    \hat{\Lambda}^\dagger & 0 \\
 \end{bmatrix}$, where we are concerned with establishing an upper bound for the maximal value $ \sigma_{1}(\hat{\Lambda})=\lambda_1(\chi)$. Since \cref{th:matrixchernoff} requires non-negative Hermitian operators, we rewrite $\chi$ as sum of non-negative operations. Let $P^{(+)}$ ($P^{(-)}$) denote the projectors in the eigenspaces of $\chi$ that are non-negative (negative). Then $\chi^{(\pm)} := \pm P^{(\pm)} \chi$  are positive semidefinite matrices such that $\chi = \chi^{(+)} - \chi^{(-)}$. Using Weyl's inequality for sums of Hermitian matrices \cite{replacethiswithatextbook}, one can substitute this Hermitian sum into the right hand side of $\sigma_{1}(\hat{\Lambda}) = \lambda_1(\chi) \leq \lambda_1(\chi^{(+)}) + \lambda_1(\chi^{(-)})$ and
 \begin{align}
     \mathbb{E} [\sigma_{1}(\hat{\Lambda})] & \leq \mathbb{E} [\lambda_1(\chi^{(+)})] + \mathbb{E} [\lambda_1(\chi^{(-)})],
 \end{align} where the last step follows from taking the expectation value of both sides. We can now call  \cref{th:matrixchernoff} twice for each term by setting $Y \equiv \chi^{(\pm)}$, and
 \begin{align}
     \mu := \lambda_{\mathrm{max}}(\mathbb{E}\chi^{(+)}) + \lambda_{\mathrm{max}}(\mathbb{E}\chi^{(-)}) \label{eqn:lemma:chernoffmu},
\end{align} for a common choice of $\theta$.\qed   
\end{lemma} 

Using the \cref{th:lemma}, we now focus on characterising two predominant approaches to combating noise: error mitigation and error correction. Both strategies assume intrinsic noise has certain properties in order to be effective which may not be satisfied in practice. We call these fundamental violations failure modes. We empirically characterize the effect of these failure modes on eigenspectra and singular spectra, and find that singular values shows sensitivity to the type of failure in a manner that eigenvalues do not. We find that quantum channels subject to failure modes may yield invalid quantum channels, but can still be studied using our techniques using only the properties of complex matrices. For superoperators of $n$ qubits, we take $d'=d^2= 4^n$.

For all empirical simulations below, we generate an ensemble of random channels with simulated noise and/or failure modes. Every instance of a random channel gives rise to both singular and eigenspectra of the channel superoperator, as well as a corresponding dilation. The true eigenvalues and the true singular values are computed for each channel superoperator within the ensemble and ordered eigen- and singular values are then averaged over the ensemble. We then compare these true values with the Chernoff fluctuation term, $\mu$, in \cref{th:lemma} by computing the expected value of the dilation $\chi$ over an ensemble of random superoperators and then computing the maximal eigenvalue of this averaged dilation. Noting that the $L\log(2d^2)$ term is constant for each $d^2 \times d^2$ channel superoperator, where $d=2^n$ for $n$ qubits, the Chernoff fluctuation $\mu$ is the only quantity sensitive to the change in channel properties. Application specific details for how an ensemble of random channels is generated are discussed below. 

\subsection{Imperfect error mitigation}

In understanding imperfect error mitigation, we focus on probabilistic error cancellation \cite{Temme2017Nov,vandenBerg2023Aug,Govia2025Mar,Gupta2024Jun} to obtain estimates of noise-free expectation values from an ensemble of noisy quantum circuits. First, we analyze spectral distributions of a perfectly error-mitigated, noisy `identity' gate where mitigation proceeds using exact information about the intrinsic noise. We repeat this analysis for a ladder of CX gates. In both cases, the noise is learned and mitigated over the full Pauli basis for $n$-qubits. We then compare this ideal scenario with ineffective mitigation in two ways: when learned information about the intrinsic noise is increasingly imperfect, and when state-decay events are added between repetitions of mitigated channels. For simplicity, we assume that intrinsic noise for an ideal Clifford gate, $U$, on $n$ qubits can be reshaped to be a Pauli channel via twirling $\mathcal{P}$, as illustratively shown in the inset of \cref{fig:mitigation}(a), which enables a sparse description of the noise amenable to learning and mitigation. 

We are interested in characterizing the superoperator associated with the mitigation of a twirled noisy gate. Information about the twirled noisy gate is learned via benchmarking experiments (coefficients $c_j$) and used to construct an appropriate distribution ($\{|c_i^{inv}|/\gamma\}_i$) over the Pauli group. Inserting a sampled Pauli from this distribution before the twirled, noisy gate constitutes as performing the mitigation step, denoted $\hat{\Lambda}^{-1}$. In practice, one only has access to an empirical estimator $\hat{\zeta}$ of the true channel $\hat{\Lambda} \hat{\Lambda}^{-1}$, 
\begin{align}
    \hat{\zeta}_\kappa &:= \mathbb{E}_\kappa[\hat{\Lambda} \hat{\Lambda}^{-1}] \\
    &= \gamma \left[ 
    \sum_{i,j=1}^{\kappa}  c_j \mathrm{sgn}(c_i^{inv}) \frac{|c_i^{inv}|}{\gamma} (P_j^* U^* P_i^* \otimes  P_j U P_i) \label{eqn:pec:zeta}
    \right].
\end{align} In the above, twirling represents insertions of random Pauli operators $P_j$ sampled uniformly over the elements of the Pauli group (index $j$) which reshapes intrinsic noise to a Pauli channel. Noise information is learned using a Pauli basis, where $i$ runs over the desired basis elements, and $c_i^{inv}$ are computed using learning experiments as detailed in \cite{vandenBerg2023Aug}. During mitigation, a random Pauli operator $P_i$ is sampled with a probability $\frac{|c_i^{inv}|}{\gamma}$ and inserted before the unitary gate $U$ in the circuit. The overall scale $\gamma$ and the $\pm$ signs denoted by $\mathrm{sgn}(c_i^{inv})$ are used to re-scale computed quantities in post-processing. For simplicity, the number of twirled samples (index $j$), and the number of samples used for mitigation (index $i$) are chosen to be equal ($\kappa$) in the equation above. The resulting superoperator representing this sampling procedure converges to the ideal unitary, $\hat{\zeta} \to \hat{U}$ in the high data limit ($\kappa \to \infty$).Further details can be found in \cref{app:mitigation:extra} or Ref.~\cite{vandenBerg2023Aug}. For our work in terms of singular values, one anticipates that perfect mitigation implements the ideal unitary, where the singular spectrum is $\sigma(\hat{U}\hat{U}^\dagger) = \{1\}$ with multiplicity of $d^2$.

For empirical calculations, we generate an ensemble of 100 error-mitigated superoperators $\hat{\zeta}_\kappa$. Each error mitigated superoperator is constructed by multiplying an ideal gate with Pauli operators sampled from a true Pauli noise model and then a Pauli sampled from the mitigation probability $\frac{|c_i^{inv}|}{\gamma}$. One may perturb the noise model prior to mitigation with strength $\epsilon$. For many layers of noisy gates, a product of independently generated mitigated superoperators is taken. One may probabilistically add state decay events between mitigated layers with strength $\alpha$. The resulting superoperator represents one instance of a random circuit and $\kappa$ random instances are averaged to yield the final error-mitigated superoperator. Under ideal conditions ($\epsilon=\alpha=0$), this final error-mitigated superoperator $\hat{\zeta}_\kappa$ converges to a valid quantum channel when $\kappa\to \infty$.

\begin{figure}
    \centering
\includegraphics[width=0.99\linewidth]{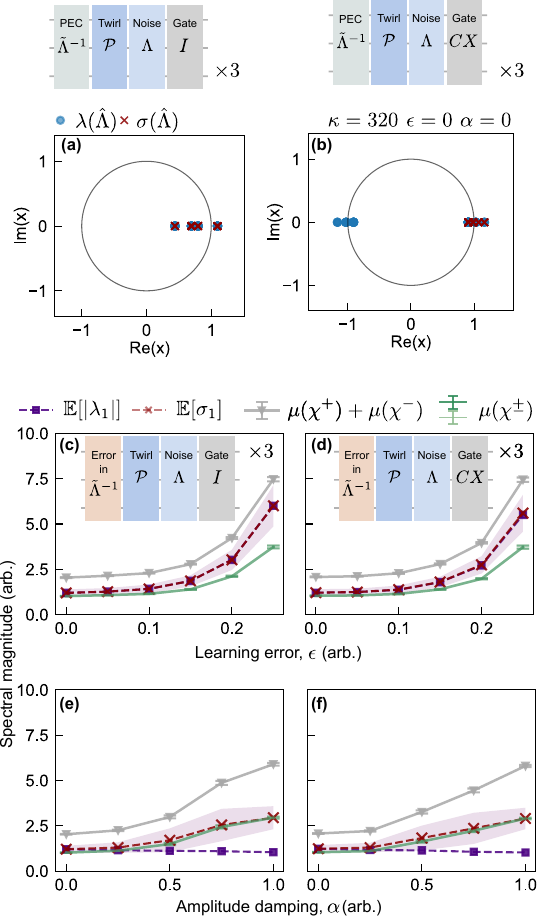}
    \caption{Spectral distributions for three layers of an error mitigated identity ($I_1I_2I_3$, left) or CX ladder ($CX_{1,2}CX_{2,3}$, right) using PEC \cite{vandenBerg2023Aug}, where the intrinsic noise model is simply set to $X,Z$ errors on the first qubit. The true eigenspectral radius $\mathbb{E}[\lambda_1]$ (indigo) and singular spectral radius $\mathbb{E}[\sigma_1]$ (red) is compared to $\mu$ (\cref{eqn:lemma:chernoffmu}).(a),(b) Instance of a perfectly mitigated random superoperator for fixed $\kappa=320$. For normal identity gate in (a), absolute eigenvalues and singular values smoothly approach unity from above and below with increasing $\kappa$. (c),(d) Under imperfect learning, the intrinsic noise coefficients are perturbed by uniformly random variates with strength proportional to $\epsilon$ before being used for mitigation. Spectral quantities diverge from unity with increasing learning error $\epsilon$ in the high data limit $\kappa=320$. (e),(f) Under amplitude damping noise with strength $\alpha$, superoperators are no longer normal and absolute eigenvalues do not coincide with singular values, even for identity gate. Chernoff spectral terms and true singular values diverge with increasing amplitude damping noise strength, $\alpha$, for fixed learning error $\epsilon=0.2$ and data samples $\kappa=320$. In all cases, Chernoff fluctuation (grey) appears to be an upper bound to a diverging spectral radius.}
    \label{fig:mitigation}
\end{figure}

Using \cref{th:lemma}, we now test the extent to which singular values of error-mitigated channel approach unity. In \cref{fig:mitigation}, we consider an error-mitigated channel on three qubits, where the ideal gate is set to be the identity ($I_1I_2I_3$, left) or a CX ladder ($CX_{1,2}CX_{2,3}$, right) in each layer, and each layer is repeated 3 times. Simulated intrinsic noise is chosen to be single-qubit $X,Z$ errors on the first qubit. The consequence of the identity gate is that the resulting mitigated superoperator is normal, so eigen- (indigo) and singular (red) values coincide in \cref{fig:mitigation}(a) (c.f. CX ladder in \cref{fig:mitigation}(b)). Here, a random instance of a mitigated map with $\kappa=320$ shows that eigenvalues are not limited to the unit circle and so the mitigated map approaches a valid channel only for infinite $\kappa$ under idealized conditions. 

We move away from perfect mitigation by considering imperfect noise information in \cref{fig:mitigation}(c),(d). Here, the exact noise coefficients are perturbed by uniformly random variates with strength $\epsilon \in [0, 0.25]$, and then renormalised, before being used to calculate $c_i^{inv}$ during mitigation. The resulting superoperator remains normal, but the spectral radius rapidly rises away from unity for increasing $\epsilon$. The Chernoff fluctuation term $\mu(\chi^+) + \mu(\chi^-)$ (gray solid) continues to provide a valid upper bound to the spectral radius, even though the uniform bound $L$ cannot be safely set to unity (as for CPTP maps) in the $L\log(2d^2)$ term as a result of this divergent behavior. 

Finally, we consider energy relaxation events in \cref{fig:mitigation}(e),(f) that may affect an error-mitigated circuit. These error relaxation events are more likely to occur when total circuit depth has runtime comparable to the shortest $T_1$ time of the noisiest qubit. We model this scenario in \cref{fig:mitigation}(d), where non-unital amplitude damping noise with strength $\alpha$ is applied to all qubits between each mitigated layer. Here, the superoperator is no longer normal, so eigenspectral and singular spectral radii are no longer the same. The eigenspectral radius is unity with increasing noise strength $\alpha$, as the mitigation superoperator is dominated by features of the non-unital noise. Meanwhile, the singular spectrum diverges with increasing $\alpha$, suggesting increasing amounts of unaddressable errors that prohibit convergence to ideal unitary (where singular values would be unity with multiplicity of $d^2$). Again, we confirm the singular spectral radius (red dashed) is upper bounded by Chernoff fluctuation term (gray). Indeed, the divergence of the singular spectral radius with $\alpha$, can be contrasted with the convergence of the eigenspectral radius to unity with $\alpha$, where the latter provides no indication about the failure of the mitigation protocol as unaddressable errors increase.

Since \cref{th:lemma} applies to any complex random matrices, we are able to examine convergence of an imperfectly mitigated quantum channel (that is not CPTP) to an ideal unitary. The $x$-axis \cref{fig:mitigation}(c)-(f) represents noise strength instead of, say, time steps of a discrete Markov chain. However, we use `convergence' to mean that the noise level per step and the total number of steps can always be adjusted in simulations such that the overall noise strength equals the $x$-axis in (c)-(f). In this sense, the sample-average and the long-time average are equivalent for this toy problem, but the same analysis could be repeated in terms of number of steps when such ergodicity is not guaranteed. 

\subsection{Imperfect error correction}

\begin{figure}
    \centering
    \includegraphics[width=0.99\linewidth]{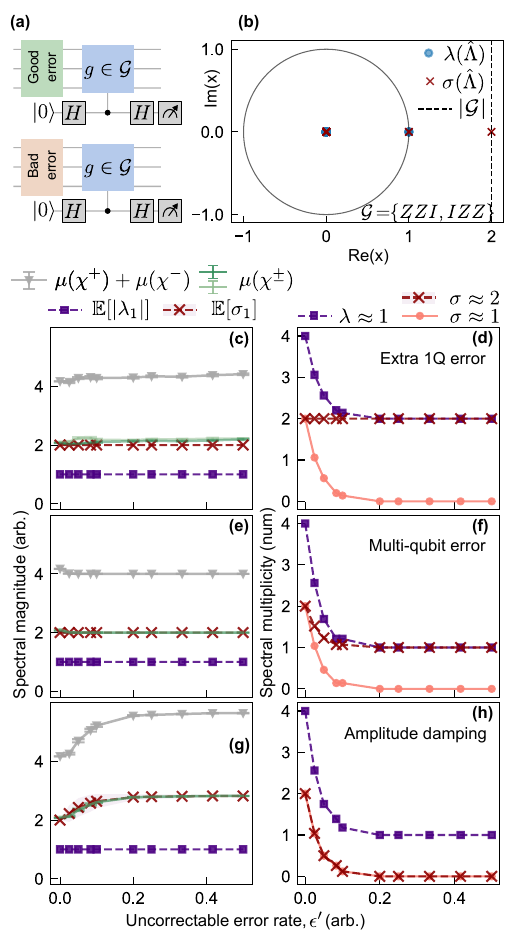}
    \caption{Spectral distributions for $d=3$ qubit bit-flip code under perfect error correction vs. imperfect error correction after $25$ rounds of $|\mathcal{G}|$ measurements. Rows correspond to 3 failure modes: extra bit-flip error, multi-qubit Pauli error, or amplitude damping error per round. (a) illustrates that errors are added before a full round of stabilizer measurements. Uncorrectable errors are added with probability $\epsilon' \in  [0, 1]$. (b) For reference, a perfect error correcting channel with correctable errors ($\epsilon'=0$) reveals eigenvalues $\in \{0,1\}$ and singular values  $\in \{0,1, |\mathcal{G}|\}$ are degenerate. (c),(e),(g) Largest singular value (red dashed) is constant for non-unital noise in (c),(e) and Chernoff fluctuation term (grey solid) appears to be a loose upper bound in all cases. (d),(f),(h) Multiplicity of singular values at $\sigma=\{1, |\mathcal{G}|\}$ decays from two to unity (zero) for Pauli (amplitude damping) noise as uncorrectable error probability $\epsilon'$ increases. Meanwhile, multiplicity for $|\lambda| \approx 1$ falls from $4^{k}$ to unity for all uncorrectable error models.  
    \label{fig:errorcorrection}}
\end{figure}

Finally, we analyze the singular spectrum for a quantum channel representing quantum error correction (QEC) code subject to Pauli noise. Perfect error correction constitutes as a code subject to a noise model for which errors are both detectable and correctable with respect to the code. In this section we characterize the effect on eigen- and singular values when noise purposefully violates assumptions of the protocol.

In quantum error correction, information is protected using redundancy by encoding quantum information into a larger set of qubits \cite{Gottesman1997May}. For so-called stabilizer codes, a logical qubit refers to degrees of freedom in subspace formed by the joint $+1$ eigenspace of a set of commuting Pauli operators, $\mathcal{G}$, called stabilizers. The specification of stabilizers essentially defines the code, and are chosen such that single-qubit Pauli errors anti commute with this set. Consequently, a logical qubit affected by errors can be detected in the $-1$ eigenspace of one or more stabilizers \cite{Gottesman1997May}. Assuming Pauli noise with sufficiently low error rates, it is sometimes possible to further apply a recovery operation to correct the effect of detected errors. In order to detect whether the logical state is impacted by errors (i.e. resides $-1$ eigenstate of one or more stabilizers), one appends an ancillary qubit to the desired system, and performs an ancilla-controlled Hadamard measurement of each stabilizer generator (see \cref{fig:errorcorrection}(a)). Denoting the stabilizers generators $S_i$ as the set $\mathcal{G}:= \{ S_i \}_{i=1}^{n-k}$, an $[n,k]$ error correcting code encoding $k$ logical qubits into $n$ physical qubits has a channel superoperator,
\begin{align}
    \hat{\Lambda} &:= \sum_{m=1}^{2^{|\mathcal{G}|}}  P_m^* \Pi_m^* \otimes P_m \Pi_m, \label{eq:recover}\\
    \Pi_m &: =  \prod_{i=1}^{|\mathcal{G}|} \frac{1}{2}\left(\mathbb{I} + (-1)^{\nu_m(i)} S_i\right), 
\end{align} where $\Pi_m$ is a projector, characterized by the single round in which all stabilizer generators are measured to yield a unique bitstring $\vec{\nu}_m := (\nu_m(1),\dots, \nu_{m}(n-k))$; $\nu_m(i)$ is a binary measurement outcome on the $i$-ancilla qubit for stabilizer $S_i$, and $P_m$ is resultant recovery operator for correctable errors \cite{Gottesman1997May, ippoliti_perturbative_2015}. As detailed in \cref{app:qec:extra} (\cref{app:qec:eigenvals}), the eigenvalues of this channel comprises of $0$ and $1$ with multiplicities $4^n-4^k$ and $4^k$ respectively. We find that singular values of the code are given by \cref{main:qec:singularvalues}.

\begin{theorem} \label{main:qec:singularvalues}
    Consider an $[n,k]$ stabilizer code with the set of stabilizer generators $\mathcal{G}=\{S_j\}_{j=1}^{n-k}$, the set of correctable Pauli errors $\{P_m\}_{l=1}^{2^{n-k}}$ with $P_1=\mathbb{I}$ and error-correcting channel given by
    \begin{equation}
        \Lambda(\rho) = \sum_{m=1}^{2^{|\mathcal{G}|}} P_m \Pi_m \rho \Pi_m P_m, 
    \end{equation} corresponding to the superoperator representation in \cref{eq:recover}. The singular-value spectrum of $\Lambda$ comprises $2^{\frac{n-k}{2}}$ and 0 with multiplicities $4^k$ and $4^n-4^k$ respectively.
   \proof See \cref{app:qec:extra}.
\end{theorem}

While the above results hold true for a (noise-free) error correcting channel, we now investigate the behavior of this channel when it is composed with noise. In \cref{fig:errorcorrection}, we consider the three-qubit bit-flip code defined by stabilizer generators $\mathcal{G}:= \{ IZZ, ZZI \}$ \cite{eczoo_stab_5_1_3}. The well-behaved noise channel $\hat{\mathcal{N}}$ is defined by normalized Kraus operators only for the single-qubit bit-flip noise $\{ IXI, XII, IIX \}$. For perfect error correction, a maximum of only one bit-flip error per stabilizer round $\Pi_m$ is permitted and there are no measurement errors that corrupt ancillary measurements. This perfectly error-correcting channel with well-behaved noise $\hat{\Lambda}\hat{\mathcal{N}}$ has a spectral distribution in \cref{fig:errorcorrection}(b) and constitutes of one full round of stabilizer measurements. Subject to these well-behaved, correctable errors, we empirically observe singular value spectrum consists of a set of only 3 unique values ($\{0,1, |\mathcal{G}|\}$), and the eigenspectra consists of either 0 or 1.

In panels \cref{fig:errorcorrection}(c),(e),(g) the singular spectral radius is estimated with concentration (grey) with increasing rate of so-called `bad errors', $\epsilon'$. These bad errors correspond to three failure modes in each full round of stabilizer measurements: extra bit-flip error per round (apply $\hat{\mathcal{N}}$ more than once), multi-qubit Pauli errors ($\hat{\mathcal{N}}$ represents the full Pauli group) and amplitude damping noise ($\hat{\mathcal{N}}$ represents amplitude damping) with probability $\epsilon'$. In all cases, we again find Chernoff fluctuation term is a loose upper bound (even without the dimensionality term). Contributions from $\mu(\chi^\pm)$ are individually comparable to the singular spectral radius. However, for error-correcting channels under unital noise in (c), (e), the singular spectral radius is unchanged even as uncorrectable error rates increase, and the radius only exceeds $|\mathcal{G}|$ for amplitude damping (non-unital) noise in (g). While error-correcting channels are not unital under unital noise, these results reveal limited value for computing concentration results for the singular spectral radius under unital noise. 

While constant spectral radius with increasing uncorrectable error rates may appear counter-intuitive in \cref{fig:errorcorrection}(c),(e), we now examine algebraic multiplicity of these degenerate spectral values in \cref{fig:errorcorrection}(d),(f),(h). Since geometric multiplicity cannot exceed algebraic multiplicity, a decay in algebraic multiplicity provide some indication of how eigenspaces in error correction decay when affected by noise. The multiplicity of $\lambda \approx 1$ eigenvalue (indigo, dashed) reduces from four to two for extra bit-flip errors in (d), but decays to its minimum possible value of unity for unital Pauli errors and non-unital noise in (f), (h). Similarly, the singular value $\sigma \approx 1$ (pink, solid) decays to its minimum possible value of zero for all noise models. In contrast, multiplicity of singular values $\sigma = |\mathcal{G}|$ (red dashed) differs in decay rates for each failure mode: for Pauli errors, it is unchanged in (d) or decays only to unity in (f), while under non-unital noise, it decays to zero as singular values begin to exceed $|\mathcal{G}|$.

We numerically replicate these results for the 5-qubit code that corrects all Pauli single-qubit errors \cite{eczoo_stab_5_1_3} in \cref{app:qec:extra}. In both the bit-flip code and the five-qubit code, we empirically confirm that total number of non-zero singular values add to $4^{n - |\mathcal{G}|}=4^{k}$ in the ideal case where the error correcting channel is composed with a Pauli channel consisting of only correctable Pauli errors. We remark that comparisons of the behavior of singular or eigenvalues (e.g. the decay in their multiplicities) could be used to test robustness of different codes to realistic operating conditions.

\section{Upper bound for arbitrary singular values \label{sec:fullversion}}
 We now extend the result of \cref{th:lemma} for maximal singular values to any ordered eigenvalue in \cref{th:mainresult}. This extension may be useful for analyzing the properties of Markov processes defined by the channel superoperator, for example, in establishing spectral gaps, deriving relaxation or mixing times. 
 \begin{figure}
     \centering
    \includegraphics{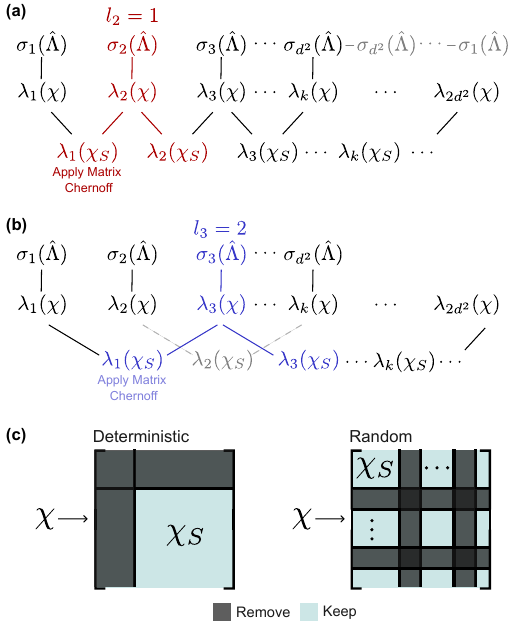}
     \caption{Interlacing property (\cref{eqn:interlacing}) for principal submatrix $\chi_S$ of dilation $\chi$ of a superoperator $\hat{\Lambda}$, in \cref{th:mainresult}, where eigen- and singular spectra are ordered from largest to smallest on the real line. Black vertical lines indicate equivalence of singular values of channels to eigenvalues of the dilation, where doubling of singular channel spectrum up to a sign is shown in gray. Slanted lines show upper and lower inequalities in \cref{eqn:interlacing}, with colored region highlighting the case for $k=1$. (a) Choosing channel singular value $i=2$ implies $l_2=1$ row/column is deleted from $\chi$ showing $\sigma_2 \in [\lambda_2(\chi_S), \lambda_1(\chi_S)]$, where $\chi_S$ is order $2d^2 -1$. (b) $i=3$ means $l_3=2$ rows/columns are deleted from $\chi$ showing $\sigma_3 \in [\lambda_3(\chi_S), \lambda_1(\chi_S)]$, where $\chi_S$ is order $2d^2 -2$. (c) Two methods by which to form a principal submatrix $\chi_S$ of $\chi$: Deterministic method systematically deletes leading rows / columns (left), while Random method uniformly samples which row / columns to delete (right).}
     \label{fig:interlace}
 \end{figure}

\begin{theorem}[Spectral concentration of singular values]\label{th:mainresult}
    Let $\hat{\Lambda} \in \mathcal{M}_{d^2}$ be a complex valued matrix with ordered singular values $ \sigma_1(\hat{\Lambda}) \geq \hdots \sigma_i(\hat{\Lambda}) \hdots  \geq \sigma_{d^2}(\hat{\Lambda})$, counted with multiplicity. Introduce constants $\theta, L \geq 0$, as in \cref{th:matrixchernoff}, as well as a function $l_i$ that returns the total number of singular values counted with multiplicity that are strictly greater than $\sigma_i$, for $i= 1, \hdots, d^2$, i.e. $l_i \in [0, d^2 -1]$, subject to the initial condition $l_j=0, \forall j $ such that $ \sigma_j = \sigma_1$. Then the $i$-th singular value, $\sigma_i(\hat{\Lambda})$ is upper bounded,
    \begin{align}
    \mathbb{E}[\sigma_{i}(\hat{\Lambda})] &\leq \mu_{i} \left( \frac{ e^\theta -1 }{\theta}\right) + \frac{2L}{\theta}\log (D_i),
    \end{align} 
    where $D_i=2d^2-l_i$, $\mu_i$ is defined by \cref{eqn:chernoffmu}, and the expectation value is computed over an ensemble of random channels.
\proof One considers the Hermitian dilation \cref{th:dilation} with $\chi : = \begin{bmatrix}
    0 & \hat{\Lambda} \\
    \hat{\Lambda}^\dagger & 0 \\
 \end{bmatrix}$, where the spectrum of $\chi$ is two copies of the singular spectrum of $\hat{\Lambda}$, denoted $\sigma(\hat{\Lambda})$, up to a sign, $\lambda(\chi) = \{ \sigma(\hat{\Lambda}), - \sigma(\hat{\Lambda})\}$, which can be observed by considering the eigenspectrum of $\chi^2$. We can order the real eigenvalues of $\chi$ and the positive spectrum will be identical to the negative spectrum. For $i=1$, we work directly with $\chi \equiv \chi_S$ but for $i>1$, we work with a principal submatrix of $\chi$. One may always extract a principal submatrix $\chi_S$ from $\chi$ by deleting $l_i$ number of rows (and correspondingly the same columns). By \cref{th:cauchyinterlace}, the eigenvalues of $\chi_S$ interlace with those of $\chi$, and a recursion gives
 \begin{align}
     \lambda_{k+l_i}(\chi_S) \leq \lambda_{k+l_i}(\chi) & \leq \lambda_{k} (\chi_S), \quad k = 1, \hdots 2d^2 - l_i. \label{eqn:interlacing}
 \end{align} In the recursion above, the index $k$ counts over the the number of (ordered) eigenvalues of $\chi_S$, i.e. $k$ has a maximal value of $2d^2 - l_i$. The index $i$ counts the ordered singular values of $\hat{\Lambda}$, while $l_i$ counts the total number of eigenvalues with multiplicity which are strictly greater than the value at $i$. If our choice of $k$ is such that $k+l_i \leq d^2$ then $\lambda_{k+l_i}(\chi)$ are positive singular values of $\hat{\Lambda}$, else if $k+l_i > d^2$, $\lambda_{k+l_i}(\chi)$ corresponds to the symmetric, negative singular values of $\hat{\Lambda}$. Focusing on the positive spectrum ($k$ such that $k + l_i \leq d^2$), singular values $\sigma_{k + l_i}(\hat{\Lambda}) \equiv \lambda_{ k + l_i}(\chi)$ are eigenvalues of $\chi$.
 
We now are concerned with establishing an upper bound for the maximum eigenvalue of $\lambda_1(\chi_S)$ (i.e. $k\equiv 1$), which is an upper bound for $\sigma_{1+l_i}(\hat{\Lambda})$. For all $i$, we see that $ \lambda_{l_i + 1}(\cdot) \equiv \lambda_{i}(\cdot)$ by the definition of $l_i$. By substitution for the positive spectrum, we see
 \begin{align}
     \lambda_{1+l_i}(\chi_S) \leq \lambda_{1+l_i}(\chi) & \leq \lambda_{1} (\chi_S),  \label{eqn:interlacing:2}\\
     \implies \lambda_{i}(\chi_S) \leq \sigma_{i}(\hat{\Lambda}) & \leq \lambda_{1} (\chi_S) \label{eqn:interlacing:3}.
 \end{align} Note that $\lambda_{1} (\chi_S)$ is the leading eigenvalue of a principal submatrix of a Hermitian matrix. Since the principal submatrices of the Hermitian matrices are also Hermitian, $\chi_S$ is Hermitian. Since \cref{th:matrixchernoff} requires non-negative Hermitian operators, we rewrite $\chi_S$ as sum of non-negative operations. Let $P^{(+)}$ ($P^{(-)}$) denote the projectors in the eigenspaces of $\chi_S$ that are non-negative (negative). Then $\chi_S^{(\pm)} := \pm P^{(\pm)} \chi$ are positive semidefinite matrices such that $\chi_S = \chi_S^{(+)} - \chi_S^{(-)}$. Using Weyl's inequality for sums of Hermitian matrices \cite{replacethiswithatextbook}, one can substitute this Hermitian sum into the right hand side of \cref{eqn:interlacing:3}, 
 \begin{align}
     \lambda_{i}(\chi_S) & \leq \sigma_{i}(\hat{\Lambda}) \leq \lambda_1(\chi_S^{(+)}) + \lambda_1(\chi_S^{(-)}), \\
    \implies \mathbb{E} [\sigma_{i}(\hat{\Lambda})] & \leq \mathbb{E} [\lambda_1(\chi_S^{(+)})] + \mathbb{E} [\lambda_1(\chi_S^{(-)})],
 \end{align} where the last step follows from taking the expectation value of both sides. We can now call \cref{th:matrixchernoff} twice for each term by setting $Y \equiv \chi_S^{(\pm)}$. For a common choice of $\theta$, one can sum the $\mu_{\mathrm{max}}$ contributions by defining 
 \begin{align}
     \mu_i := \lambda_{\mathrm{max}}(\mathbb{E}\chi_S^{(+)}) + \lambda_{\mathrm{max}}(\mathbb{E}\chi_S^{(-)}) \label{eqn:chernoffmu},
\end{align} where dependence of $\mu_i$ on $i$ alludes to the $2d^2-l_i$ order of submatrices.\qed  
\end{theorem} 

We provide examples to clarify the interlacing step in \cref{eqn:interlacing}. 

\textit{Example 1 (interlacing non-degenerate singular values)}: Suppose $\hat{\Lambda}$ has non-degenerate singular values with $d^2=4$ and $\chi$ is an $8 \times 8$ matrix. The multiplicity function $l_i = i -1 $  cannot return a value greater than $d^2-1$ i.e., $l_1 = 0, l_2 = 1, l_3 = 2, l_4 = 3$. Suppose we pick $i=2$ to upper bound the second largest singular value. Then we truncate $l_2=1$ row/column from $\chi$, yielding a submatrix $\chi_S$ with row/column indices ranging from $k=1, \hdots, 7 $ in the recursion equation $\lambda_{k+1}(\chi_S) \leq \lambda_{k+1}(\chi) \leq \lambda_{k} (\chi_S)$. Since the Chernoff bound applies only to the maximal eigenvalues, we proceed by setting $k=1$ to get $\lambda_{2}(\chi_S) \leq \lambda_{2}(\chi) \leq \lambda_{1} (\chi_S)$ via interlacing (see \cref{fig:interlace}(a)).

\textit{Example 2 (interlacing degenerate singular values)}: Suppose $\hat{\Lambda}$ has degenerate singular values, e.g. with $d^2=4$, let $\sigma_1 = 1, \sigma_2 = 1, \sigma_3 = 0.9, \sigma_4 =0$. The multiplicity function now gives $l_1 = 0, l_2 = 0, l_3 = 2, l_4 = 3$. Picking $i=2$ yields no truncation, and Chernoff bound is applied only to the spectral radius of $\chi \equiv \chi_S$. To upper bound the second largest singular value instead, one must pick $i=3$. Then we truncate $l_3=2$ row/columns from $\chi$, yielding a submatrix $\chi_S$ with row/column indices ranging from from $k=1, \hdots, 6 $ in the recursion equation $\lambda_{k+2}(\chi_S) \leq \lambda_{k+2}(\chi) \leq \lambda_{k} (\chi_S)$. Again, the Chernoff bound applies only to the extremal eigenvalues, we proceed by setting $k=1$ to get $\lambda_{3}(\chi_S) \leq \lambda_{3}(\chi) \leq \lambda_{1} (\chi_S)$ via interlacing. Meanwhile choosing $k=2$ gives interlacing for the even-numbered singular values, $\lambda_{4}(\chi_S) \leq \lambda_{4}(\chi) \leq \lambda_{2} (\chi_S)$. (See \cref{fig:interlace}(b)).

\textit{Example 3 (no interlacing)}: Suppose $\hat{\Lambda}$ has all degenerate singular values e.g. $\hat{\Lambda}^\dagger \hat{\Lambda} = I$ and singular value of $1$ has multiplicity of $d^2$. Then $l_i = 0, \quad \forall i$ and no non-trivial truncation of $\chi$ exists.

We make some remarks on free parameters and limitations of the proof. Firstly, it is not clear how the principal submatrix should be extracted before applying the Cauchy interlace theorem. For example, one strategy is to extract $l_i$ number of columns/rows from the top left of every matrix (`Deterministic'), while another strategy is to randomly select which rows/columns need to be deleted (`Random'), see \cref{fig:interlace} (c) for an illustration. We find that the latter strategy appears to lead to a tighter matrix Chernoff upper bound in applications. Secondly, a limitation of our proof is that it cannot be used to find the matrix Chernoff lower bound for the minimal eigenvalue of $\chi$, since we decompose a Hermitian matrix into a sum of positive semidefinite terms in the final steps. The decomposition of a Hermitian matrix to positive semidefinite components is necessary to invoke the matrix Chernoff concentration result. Thirdly, like classical Chernoff matrix inequalities, there is some freedom to control sharpness of these bounds for specific applications e.g. by choosing free parameters $\theta, \epsilon$. However, in many applications, the Chernoff bounds are found to be non-sharp for matrices with decaying spectra due to the $L\log(D)$ term. This issue arises as the dimension, $D_i$, multiplied by the maximum eigenvalue, is used to upper bind the trace to derive \cref{th:matrixchernoff} and leads to non-sharp bounds for decaying spectra. Proposed improvements include replacing the dimensional factor $D_i$ by a smaller `intrinsic' dimension, while in other cases, it may be possible to omit $\frac{L}{\theta}\log{D_i}$ on an application-specific basis \cite{tropp_introduction_2015}. Finally, the proof does not claim a computational advantage. One may directly estimate sample averages and variances of eigen- or singular spectral values for a random matrix ensemble. However, one can additionally guarantee spectral concentration by using \cref{th:mainresult}.

From an applications perspective, we note that the Hermitian dilation preserves the Schatten infinity norm, or the spectral norm, $||\chi||_\infty = ||\hat{\Lambda}||_\infty = ||\Lambda||_{2,2} =\lambda_1(\chi)$, where second last inequality occurs by definition (cf. equation 2.14 in Ref. \cite{fukuda_concentration_2024}) and the last equality represents the spectral radius of $\chi$ \cite{tropp_introduction_2015}. It is left to future work how concentration of ordered singular values for random superoperators in \cref{th:mainresult} could help in defining statistical properties of induced or completed bounded Schatten $\infty$-norms that are typically used to understand distance between superoperators of channels. Indeed our results might be useful in deriving bounds for estimating induced Schatten norms \cite{JohnNorms2025Jul} when channel superoperators are random. Another set of applications of concentration of singular values could be found in characterizing relaxation times and spectral gaps of classical Markov chains, for instance, for irreversible chains \cite{Wolfer2022Sep,chatterjee_spectral_2025}. Typically one seeks an upper (lower) bound for the second-largest (largest) singular value for non-normal or non-square channels to establish the existence of a gap. While determining when singular values of classical Markov chains may apply to quantum Markov processes is beyond the scope of present work, we provide numerical evidence that\cref{th:mainresult} can be used to estimate spectral gaps for the simpler case of ergodic/mixing quantum Markov processes in \cref{app:ergodicqc}.

\section{Conclusion \label{sec:conclusion}}

In this manuscript, we propose a theoretical framework to analyze singular values of ensembles of random superoperators representing imperfect quantum protocols. We use matrix Chernoff inequalities to upper bind singular values of random operators over this ensemble and find that singular spectra can show signatures of different failure modes of quantum computation under realistic operating conditions. Our theoretical framework is generally applicable to ensembles of random complex matrices, and therefore can be used to study random superoperators that converge to valid quantum channels only under the idealized limit that failure probability goes to zero.

Our proposed framework is additionally used to characterize singular values for imperfect error mitigation and error correction. For all error mitigation applications where ideal quantum operations are unitary, we empirically confirm that singular values converge to unity from both above and below as the number of samples increase. For failure modes related to both learning error and non-unital noise, we find that Chernoff concentration provides upper bounds that scale with the strength of unaddressable errors in the protocol. For error correction, we provide a characterization of the singular spectrum in the absence of noise, where magnitude and multiplicities of singular values are linked to degrees of freedom in $[n,k]$ codes. We find that Chernoff upper bound for maximal singular values grows with increasingly unaddressable error only for non-unital noise. Under unital noise, the Chernoff upper bound is constant but singular value multiplicities are seen to decay with noise strength. 

These results suggest that singular values may provide information that could be linked to operational metrics for trusting the output of noisy quantum computers when protocols for combating noise are imperfect. Our work paves the way for new theoretical tools that can accommodate a wider range of how quantum computation may fail on hardware. Future research directions could include looking at the impact of our work on channel norms, approximation errors, or establishing spectral gaps for families of quantum Markov processes. All of these avenues are exciting directions for studying spectral concentration in channels.\\
\\
\textbf{Acknowledgments:} S.K., K.B., L.J.H and S.S. are supported by the ARC Centre of Excellence for Engineered Quantum Systems (CE17010000) with S.K., K.B. additionally supported via Deborah Jin Fellowships. R.S.G. is supported by UQ's Queensland Digital Health Center via funding from UQ’s Health Research Accelerator (HERA) initiative. Authors thank Luke Govia and Hakop Pashayan for useful discussions.\\
\textbf{Competing interests:} All authors declare no financial or non-financial competing interests.\\
\textbf{Data and code availability:} All scripts and data are available upon reasonable request and will be made publicly available at time of publishing.\\
\clearpage
\onecolumngrid
\appendix

\section*{Supplemental information}

\section{ Quantum channels \label{app:channels}} 
A quantum channel is a completely positive map represented by the following operator sum, 
    \begin{align}
    \Lambda (\rho) := \Sigma_{i=1}^\kappa A_i \rho A_i^\dagger, \quad \Sigma_i A_i^\dagger A_i \leq \mathbb{I}, \label{eqn:krausrep2}
\end{align} where $\rho \in \mathcal{S}_d$ is an input quantum state, and $A_i \in \mathcal{M}_d $ are Kraus operators for a total number of $\kappa$ terms in the sum. The channel is trace preserving $ \iff \Sigma_i A_i^\dagger A_i = \mathbb{I}$ or unital $ \iff \Sigma_i A_i A_i^\dagger = \mathbb{I}$, where $\mathbb{I}$ is the identity. The channel $\Lambda$ has a matrix representation $\hat{\Lambda} \in \mathcal{M}_{d^2}$, with matrix elements given by $\hat{\Lambda}_{ij}:= \langle B_i |\Lambda |B_j\rangle$ for an orthonormal basis $\{B_i\}$ of $\mathcal{M}_{d^2}$, i.e. the Hilbert-Schmidt norm in regular operator space. This basis will refer to matrix units in this work. One obtains superoperator representation from \cref{eqn:krausrep2} as 
   \begin{align}
        \hat{\Lambda} := \Sigma_{i=1}^\kappa  A_i^* \otimes A_i,  \label{eqn:superoprep2}
    \end{align} where ${}^*$ denotes entry-wise conjugation.
The eigenvalues of a complex square matrix $\hat{\Lambda} \in \mathcal{M}_{d^2}$ are the set of complex numbers that satisfy $\mathrm{det}(\lambda\mathbb{I}_{d^2} - \hat{\Lambda})$. We use $\lambda(X)$ to denote the full set of eigenvalues of the operator $X$, as well as to denote the $i$-th ordered eigenvalue $\lambda_i(X)$. 

For superoperators $\hat{\Lambda}$ that are completely positive and trace-preserving (CPTP), it is known that the eigenspectra are contained within the unit disc ($ \forall \lambda \in \lambda(\hat{\Lambda}), \quad |\lambda| \leq 1$) \cite{wolf_tour_2012, watrous_theory_2018}. Eigenvalues with an absolute magnitude of unity, $|\lambda|=1$, are called peripheral eigenvalues. In particular, a superoperator representing a quantum CPTP map is similar to its Jordan normal form $T$ for some unitary $U$, where each Jordan block is decomposed into diagonal and nilpotent matrices. From \cite{wolf_tour_2012}, one may decompose a superoperator for a CPTP map as 
\begin{align}
    \hat{\Lambda} = U T U^\dagger, \quad T: &= \sum_k \lambda_k P_k + N_k, \label{eqn:jordandecomp}
\end{align} where $\lambda_k$ is the $k$-the eigenvalue with a Jordan block given by $\lambda_k P_k + N_k$, $P_k$ are projectors satisfying $P_k^2 = P_k$, $\sum_k P_k = \mathbb{I}$, $P_j N_k = \delta_{k,j}N_k$ and $P_j P_k = \delta_{k,j}P_k$. Here, it is not required that $P_k = P_k^\dagger$. Within this peripheral spectrum, there exists at least one eigenvalue $\lambda_1=1$, and that the eigenspace of $\lambda_1=1$ contains a positive semi-definite element i.e. the span of potentially nonphysical quantum states contains at least one valid quantum state\cite{wolf_tour_2012, kukulski_generating_2021}. For superoperators of CPTP maps, peripheral eigenvalues have trivial Jordan blocks, namely, $N_k =0$ for $|\lambda_k| = 1$ and the projection to the peripheral spectrum is a completely positive map given by $\sum_{k: |\lambda_k| = 1} P_k$ \cite{wolf_tour_2012}. This term allows us to define a projector that projects us into the eigenspace of the nilpotent spectrum as, 
\begin{align}
   \mathcal{N} := \mathbb{I}  -\sum_{k: |\lambda_k| = 1} P_k. \label{eqn:nilpotent}
\end{align} We mention in passing that peripheral eigenspectra and their eigenspaces helps to characterize the behavior of quantum channels as classical discrete-time Markov chains. For example, a channel is ergodic if $\lambda_1=1$ has an eigenvector that corresponds to a unique, positive definite quantum state. Of ergodic channels, a channel is further mixing if $\lambda_1=1$ is the only peripheral eigenvalue with multiplicity one \cite{burgarth_ergodic_2013}. 

Meanwhile, singular values of quantum channels, $\sigma_i$ of $\hat{\Lambda}$ are the square-root of eigenvalues of the products,
\begin{align}\label{eq:channelsv}
    \sigma_i(\hat{\Lambda}) := \sqrt{\lambda_i(\hat{\Lambda}^\dagger \hat{\Lambda})} = \sqrt{\lambda_i(\hat{\Lambda} \hat{\Lambda}^\dagger)}, \quad \forall i.
\end{align} The right (left) singular vectors are the eigenvectors of $\hat{\Lambda}^\dagger \hat{\Lambda}$ ($\hat{\Lambda} \hat{\Lambda}^\dagger$). We use $\sigma(X)$ to denote the full set of singular values of operator $X$, and the $i$-th ordered singular in this set as $\sigma_i(X)$. When $\hat{\Lambda}$ is normal ($[\hat{\Lambda}, \hat{\Lambda}^\dagger]=0$), then singular values reduce to the absolute values of $\hat{\Lambda}$. Normal operators include Hermitian, unitary, and projection operators, but the superoperator representing channels need not be normal generally. For completely-positive trace-preserving maps, the maximum singular value can be between unity and the dimension of the Hilbert space, and is exactly unity for unital channels \cite{watrous_theory_2018}. For random CPTP superoperators that may always be written in Kraus form, maximum singular values are unity by renormalization of Kraus terms to satisfy trace preservation, namely, that the dual channel is unital (cf. random Kraus maps in Ref.~\cite{kukulski_generating_2021}).

\section{ Error mitigated channels \label{app:mitigation:extra}} 

Error mitigation using probabilistic error cancellation \cite{vandenBerg2023Aug,Gupta2024Jun} obtains estimates of noise-free expectation values from an ensemble of noisy quantum circuits. As illustratively shown in the inset of \cref{fig:mitigation}(a), it is assumed that intrinsic noise for an ideal Clifford gate, $U$, on $n$ qubits can reshaped to be a Pauli channel via twirling $\mathcal{P}$. The twirled noisy gate can be represented as,
\begin{align}
   \hat{\Lambda} := \sum_{j=1}^{4^n} c_j (P_j^* U^* \otimes  P_j U),
\end{align} where noise information $c_j$ after twirling is learned via benchmarking experiments, and used to construct an appropriate distribution $\{|c_i^{inv}|/\gamma\}_i$ over the Pauli group.  The factor, $\gamma:= \sum_{i=1}^{4^n} |c_i^{inv}|$, is used to normalize the magnitudes of the coefficients, where the index $i$ runs over the Pauli group elements, and $c_i^{inv}$ are computed using noise learning experiments as detailed in Ref.~\cite{vandenBerg2023Aug}. Inserting a sampled Pauli from this distribution,
\begin{align}
    \hat{\Lambda}^{-1} := \gamma \sum_{i=1}^{4^n}\mathrm{sgn}(c_i^{inv}) \frac{|c_i^{inv}|}{\gamma} P_i^* \otimes P_i,
\end{align} before the noisy gate constitutes as performing the mitigation step. In the above, a random Pauli operator $P_i$ is sampled with a probability $\frac{|c_i^{inv}|}{\gamma}$ and inserted before the unitary gate $U$ in the circuit. The overall scale $\gamma$ and the $\pm$ signs denoted by $\mathrm{sgn}(c_i^{inv})$ are used to re-scale computed quantities in post-processing. In practice, one only has access to an empirical estimator of the true channel $\hat{\Lambda} \hat{\Lambda}^{-1}$, where the estimator averages over the number of samples in the mitigation ensemble, as well as the number of samples used during twirling. Sample sizes depend on heuristic analysis on the amount of data required for reaching a desired low-variance in noise-free expectation values. We denote this estimator $\hat{\zeta}$, 
\begin{align}
    \hat{\zeta}_\kappa &:= \mathbb{E}_\kappa[\hat{\Lambda} \hat{\Lambda}^{-1}] \\
    &= \gamma \left[ 
    \sum_{i,j=1}^{\kappa}  c_j \mathrm{sgn}(c_i^{inv}) \frac{|c_i^{inv}|}{\gamma} (P_j^* U^* P_i^* \otimes  P_j U P_i)
    \right],
\end{align} where the average over $\kappa$ represents the number of random circuits used to compute the combined superoperator for a mitigated, noisy unitary. For simplicity, the number of twirled samples (index $j$), and the number of samples used for mitigation (index $i$) are chosen to be equal ($\kappa$). The superoperator representing this sampling procedure converges to the ideal unitary, $\hat{\zeta} \to \hat{U}$ in the high data limit ($\kappa \to \infty$). In terms of singular values, one anticipates that perfect mitigation implements the ideal unitary, where the singular spectrum is $\sigma(\hat{U}\hat{U}^\dagger) = \{1\}$ with multiplicity of $d^2$.

\section{ Error corrected channels \label{app:qec:extra}} 

In this section, we first prove two results
about the spectrum of a perfect $[n,k,d]$ error-correcting channel without noise. The first results concerns the eigenspectrum of the perfect error-correcting channel whereas the second one characterizes the singular-value spectrum of the same. Next we empirically characterize the singular and eigenspectrum of QEC channels under noise for both the three-qubit repetition code presented in the main text, as well as the five-qubit code for correcting any single-qubit Pauli error. 

Before proving theoretical results, we first discuss some properties of a qubit-stabilizer code that will be used in the proofs. Notice that any $[n,k]$ code divides the $2^n$-dimensional Hilbert space into $2^{n-k}$ orthogonal subspaces of dimensions $2^k$ each.
One can find $2^{n-k}$ Pauli operators $\{P_m\}$ with $P_1:=\mathbb{I}$ such that for each stabilizer generator $S_j\in \mathcal{G}$, we have
\begin{equation}\label{eq:correrrchi}
    P_mS_j = (-1)^{\nu_m(j)}S_jP_m,
\end{equation}
where the binary string $\vec{\nu}_m := (\nu_m(1),\dots, \nu_{m}(n-k))$ is distinct for each $P_m$.
For a given $P_m$, $\vec{\nu}_m$ also encodes the measurement outcomes if the stabilizer generators $S_j$ are measured after an encoded state is acted upon by the error $P_m$. Note that $P_1=\mathbb{I}$ represents the error-free case where all the stabilizer measurements return 0 and the encoded state does not leave the code space. Furthermore, we can show that the $2^{n-k}$ orthogonal subspaces can be defined as
\begin{equation}\label{eq:cospace}
    \Pi_m := \frac{1}{2^{n-k}}\prod_{j=1}^{n-k}\left[ I+(-1)^{\nu_m(j)}S_j\right], \quad m\in\{1,\dots,2^{n-k}\}. 
\end{equation}
In other words, an error $P_m$ takes an encoded state from the code space $\Pi_1$ to $\Pi_m$.
One can construct a perfect error-correcting channel $\Lambda$ from $\{P_m\}$ as given in Eq.~\eqref{eq:recover}. The recovery operation corresponding to the syndrome $\vec{\nu}_m$ is $P_m$. Given such a $\Lambda$, the set $\{P_m\}$ can be referred to as the set of correctable errors of the corresponding QEC code as it
is easy to verify that any error from this set will be perfectly corrected after performing a round of error correction.

Let $\{\ket{1,i}\}_{i=0}^{2^k-1}$ define an orthonormal basis for the code space characterized by $\Pi_1$.
As a consequence of the fact $P_m$ and $P_{m'}$ take a state in $\Pi_1$ to mutually orthogonal subspaces $\Pi_m$ and $\Pi_{m'}$ respectively,
it follows that 
\begin{equation}\label{eq:cospaceortho}
    \bra{1,i'}P_{m'}P_m\ket{1,i} = \delta_{m,m'}\delta{i,i'},
\end{equation}
Therefore, the set 
\begin{equation}\label{eq:cospacebasis}
    \{ \ket{m,i}:= P_m\ket{1,i}\}
\end{equation}
defines an orthonormal basis for the full $2^n$-dimensional Hilbert space and
\begin{equation}\label{eq:projexp}
    \Pi_m = \sum_{i_1,i_2=1}^{2^k}\ket{m,i_1}\bra{m,i_2},\quad \forall\, m\in \{1,\dots, 2^{n-k}\}.
\end{equation}

Having established the above properties of a stabilizer code, we now prove the first result of the section.

\begin{theorem} \label{app:qec:eigenvals}
    Consider an $[n,k]$ stabilizer code with the set of stabilizer generators $\mathcal{S}=\{S_j\}_{j=1}^{n-k}$, the set of correctable errors $\{P_m\}_{m=1}^{2^{n-k}}$ with $P_1=\mathbb{I}$ and the perfect error-correcting channel
    \begin{equation}
        \Lambda(\rho) = \sum_m P_m^\dagger \Pi_m \rho \Pi_m P_m.
    \end{equation}
   The eigenspectrum of $\Lambda$ comprises 0 and 1 with multiplicities $4^n-4^k$ and $4^k$ respectively.
\end{theorem}
\begin{proof}
    First we show that $\Lambda$ satisfies $\Lambda\circ \Lambda = \Lambda$. This implies that the eigenvalues of $\Lambda$ are either 0 or 1. To see this we first note that $\Lambda$ can be rewritten as
    \begin{equation}\label{eq:recoveryalt}
        \Lambda(\rho) = \sum_m \Pi_1 P_m\rho  P_m\Pi_1
    \end{equation}
    using Eqs.~\eqref{eq:correrrchi} and~\eqref{eq:cospace} as well as the fact that $P_m$s are Hermitian.  Expand the projector $\Pi_1$ using Eq.~\eqref{eq:projexp}
to observe    
\begin{align}
        \Lambda\circ \Lambda (\rho) &= \sum_{m_1,m_2}\Pi_1P_{m_1} \left(\Pi_1P_{m_2} \rho P_{m_2}\Pi_1 \right) P_{m_1}\Pi_1 \nonumber \\&=\sum_{m_1,m_2}\sum_{i_1,i_2,i_3,i_4}\ket{1,i_1}\bra{1,i_1}P_1P_{m_1}\ket{1,i_2}\bra{1,i_2}P_{m_2}\rho P_{m_2}\ket{1,i_3}\bra{1,i_3}P_1P_{m_1}\ket{1,i_4}\bra{1,i_4} \nonumber\\
        &= \sum_{m_1,m_2}\sum_{i_1,i_2,i_3,i_4}\delta_{m_1,1}\delta_{i_1,i_2}\ket{1,i_1}\bra{1,i_2}P_{m_2}\rho P_{m_2}\delta_{m_1,1}\delta_{i_3,i_4}\ket{1,i_3}\bra{1,i_4} = \sum_{m_2}\sum_{i_1,i_3}\ket{1,i_1}\bra{1,i_1}P_{m_2}\rho P_{m_2}\ket{1,i_3}\bra{1,i_3}\nonumber\\
        &= \sum_{m_2}\Pi_1P_{m_2}\rho P_{m_2}\Pi_1 = \Lambda(\rho),
    \end{align}
    as desired. We have used the orthogonality relation given in Eq.~\eqref{eq:cospaceortho} in deriving the above.
    
    Now we compute the multiplicities of the 0 and 1 eigenvalues, respectively. 
    We do this by constructing a linearly independent set of eigenbasis for $\Lambda$ that spans the entire bounded-operator space.
    For this, we recall from Eq.~\eqref{eq:cospacebasis} that $\{\ket{m,i}\}$ forms an orthonormal basis of the full Hilbert space.
    Therefore, $\{\ket{m_1,i_1}\bra{m_2,i_2}\}$ form a basis for the corresponding bounded-operator space. We now construct the eigenvectors corresponding to the different eigenvalues of $\Lambda$. For $2^k$ basis vectors of the form $\{\ket{1,i_1}\bra{1,i_2}\}$ with support in the code space, we have
    \begin{align}
        \Lambda(\ket{1,i_1}\bra{1,i_2}) &= \sum_m \Pi_1 P_m \ket{1,i_1}\bra{1,i_2} P_m \Pi_1 = \sum_m\sum_{i_3,i_4}\ket{1,i_3}\bra{1,i_3}P_m P_1\ket{1,i_1}\bra{1,i_2}P_1 P_m \ket{1,i_4}\bra{1,i_4} \nonumber \\
        &= \sum_m\sum_{i_3,i_4}\delta_{i_1,i_3}\delta_{1,m} \delta_{i_2,i_4}\delta_{1,m}\ket{1,i_3}\bra{1,i_4} = \ket{1,i_1}\bra{1,i_2}.
    \end{align}
    That is, these $2^k$ basis vectors belong to the +1 eigenspace of $\Lambda$. Next consider a basis vector of the form $\{\ket{m_1,i_1}\bra{m_2,i_2}\}$ with $m_1\neq m_2$ and examine $\Lambda$'s action on this operator,
    \begin{align}
        \Lambda(\ket{m_1,i_1}\bra{m_2,i_2}) &= \sum_m \Pi_1 P_m \ket{m_1,i_1}\bra{m_2,i_2} P_m \Pi_1 = \sum_m\sum_{i_3,i_4}\ket{1,i_3}\bra{1,i_3}P_m P_{m_1}\ket{1,i_1}\bra{1,i_2}P_{m_2} P_m \ket{1,i_4}\bra{1,i_4} \nonumber \\
        &= \sum_m\sum_{i_3,i_4}\delta_{i_1,i_3}\delta_{m_1,m} \delta_{i_2,i_4}\delta_{m_2,m}\ket{1,i_3}\bra{1,i_4} = \bm{0},
    \end{align}
    as $\delta_{m_1,m} \delta_{m_2,m}$  is always zero because $m_1\neq m_2$. 
    Finally we consider basis vectors of the form  $\{\ket{m,i_1}\bra{m,i_2}\}$ for $m\neq 0$. Here we notice
     \begin{align}
        \Lambda(\ket{m,i_1}\bra{m,i_2}) &= \sum_{m_1} \Pi_1 P_{m_1} \ket{m,i_1}\bra{m,i_2} P_{m_1} \Pi_1 = \sum_{m_1}\sum_{i_3,i_4}\ket{1,i_3}\bra{1,i_3} P_{m_1}P_m\ket{1,i_1}\bra{1,i_2}P_{m} P_{m_1} \ket{1,i_4}\bra{1,i_4} \nonumber \\
        &= \sum_{m_1}\sum_{i_3,i_4}\delta_{i_1,i_3}\delta_{m_1,m} \delta_{i_2,i_4}\delta_{m_1,m}\ket{1,i_3}\bra{1,i_4} =  \ket{1,i_1}\bra{1,i_2}.
    \end{align}
    This means the operator $\ket{m,i_1}\bra{m,i_2}-\ket{1,i_1}\bra{1,i_2}$ for $m\neq 1$ also belongs to the null space of $\Lambda$. Thus the set of linearly independent basis
    \begin{equation}
        \{\ket{1,i_1}\bra{1,i_2}, \ket{m_1,i_1}\bra{m_2,i_2} , (\ket{m,i_1}\bra{m,i_2}-\ket{1,i_1}\bra{1,i_2})\}, i_1,i_2\in\{1,\dots 2^k\}, m,m_1,m_2 \in\{1,\dots 2^{n-k}\},m_1\neq m_2,m\neq 0,
    \end{equation}
 also forms an eigenbasis of $\Lambda$.  
    We conclude that the +1 eigenvalue has a multiplicity of $4^k$ and the 0 eigenvalue has a multiplicity of $4^n-4^k$.
\end{proof}
Next we characterize the singular-value spectrum of $\Lambda$.

\begin{theorem} \label{app:qec:singularvalues}
    Consider an $[n,k]$ stabilizer code with the set of stabilizer generators $\mathcal{S}=\{S_j\}_{j=1}^{n-k}$, the set of correctable errors $\{P_m\}_{m=1}^{2^{n-k}}$ with $P_1=\mathbb{I}$ and the perfect error-correcting channel
    \begin{equation}
        \Lambda(\rho) = \sum_m P_m^\dagger \Pi_m \rho \Pi_m P_m.
    \end{equation}
   The singular-value spectrum of $\Lambda$ comprises the values $2^{\frac{n-k}{2}}$ and 0 with multiplicities $4^k$ and $4^n-4^k$ respectively.
\end{theorem}

\begin{proof}
Singular values of the channel $\Lambda$ are given by the positive square root of the eigenvalues of the map $\Delta$, where $\hat{\Delta} := \hat{\Lambda}^\dagger \hat{\Lambda}$ following Eq.~\ref{eq:channelsv}.
This proof proceeds by first showing that $\Delta$ can be rescaled to some $\bar{\Delta} := \frac{1}{2^{n-k}}\Delta$, which is a Hermitian, unital channel that is also idempotent. 
This means $\bar{\Delta}$ has only 0, 1 in its eigenspectrum and no defective eigenvalues. Consequently, the eigenvalues of $\Delta$ are $2^{n-k}$ and 0. 
Finally, the fact that  $\bar{\Delta}$ is unital also lets us count the multiplicities of these eigenvalues thereby completing the proof.

Given the alternative form of $\Lambda$ in Eq.~\eqref{eq:recoveryalt}, the action of $\Delta$ is given by
\begin{equation}
    \Delta(\rho) = \sum_{m_1,m_2=1}^{2^{n-k}}P_{m_1}\Pi_1P_{m_2}\rho P_{m_2}\Pi_1 P_{m_1}.
\end{equation}
The rescaled channel $\bar{\Delta}$ is unital as
\begin{align}
    \bar{\Delta}(\mathbb{I})& = \frac{1}{2^{n-k}}\sum_{m_1,m_2}P_{m_1}\Pi_1P_{m_2}\mathbb{I} P_{m_2}\Pi_1 P_{m_1} = \frac{1}{2^{n-k}}\sum_{m_1,m_2}P_{m_1}\Pi_1\Pi_1 P_{m_1}  =  \frac{1}{2^{n-k}}\sum_{m_1,m_2}P_{m_1}\Pi_1 P_{m_1} \nonumber \\
    &=  \sum_{m_1,i} P_{m_1}\ket{1,i}\bra{1,i} P_{m_1} =\sum_{m_1,i} \ket{m_1,i}\bra{m_1,i}  = \mathbb{I},
\end{align}
where $\{ \ket{1,i} \}$and $\{ \ket{m,i} \}$  form an orthonormal basis for the code space and the full Hilbert space, respectively (see Eqs.~\eqref{eq:cospacebasis} and~\eqref{eq:projexp}).

We now characterize the  +1 eigenspace of $\bar{\Delta}$.  To this end we recall the result that any fixed point of a unital channel should commute with all its Kraus operators~\cite[Theorem~4.25]{watrous_theory_2018}.
That is, any operator $M$ such that $\bar{\Lambda}(M) = M$ satisfies
\begin{equation}
    [ P_{m_1}\Pi_1P_{m_2}, M] = \bm{0},\quad \forall\, m_1,m_2 \in\{1,\dots 2^{n-k}\}.
\end{equation}
This condition leads to the following constraints on $M$
\begin{align}
    \bra{m_1,i}M\ket{m_1,j} &= \bra{m_2,i}M\ket{m_2,j},\quad \forall\, m_1,m_2,i,j \\
    \bra{m_1,i}M\ket{m_2,j}& = 0,\quad \forall\, m_1,m_2,i,j,\quad m_1\neq m_2.
\end{align}
Consequently, $M$ takes the form $M = \oplus_l \bar{M}$, where $\bar{M}$ is an arbitrary $2^{n-k}\times 2^{n-k}$ matrix. That is, $M$ is block diagonal with identical blocks when expressed in the basis $\{\ket{m,i}\}$.
The +1 subspace of $\bar{\Delta}$ is spanned by operators of the form above and has dimensionality equal to $4^k$.

Next we show that the only remaining eigenvalue of $\bar{\Delta}$ is 0. We observe that 
\begin{equation}
\bar{\Delta}\circ \bar{\Delta} (\rho) 
= \frac{1}{2^{2(n-k)}}\sum_{m_1,m_2,m_3,m_4}P_{m_3}\Pi_1P_{m_4}\left(P_{m_1}\Pi_1P_{m_2}\rho P_{m_2}\Pi_1 P_{m_1}\right) P_{m_4}\Pi_1P_{m_3},
\end{equation}
and
\begin{equation}
    \Pi_1P_{m_4}P_{m_1}\Pi_1 = \delta_{m_1,m_4}\Pi_1,
\end{equation}
following Eq.~\eqref{eq:cospaceortho}.
Combining the above two relations, we get
\begin{equation}
    \bar{\Delta}\circ \bar{\Delta} (\rho) 
= \frac{1}{2^{2(n-k)}}\sum_{m_1,m_2,m_3,m_4}\delta_{m_1,m_4}P_{m_3}\Pi_1P_{m_2}\rho P_{m_2}\Pi_1P_{m_3} = \frac{1}{2^{n-k}}\sum_{m_2,m_3}P_{m_3}\Pi_1P_{m_2}\rho P_{m_2}\Pi_1P_{m_3} = \bar{\Delta}(\rho),
\end{equation}
as needed. Therefore, the remaining eigenvalue of $\bar{\Delta}$ is 0 with multiplicity $4^n-4^k$.
\end{proof}

We now introduce noise empirically to examine singular values of the three-qubit code and the five-qubit error correcting codes, under the limiting case of maximal noise strengths. The asymptotes for singular value multiplicities for an ensemble of random channels in Fig 3. (d), (f) and (h) in the main text are compared with deterministic channels ($\epsilon' \equiv 1$). Further, singular value multiplicities are examined for the  five qubit ($[5,1,3]$) code which, unlike the three-qubit ($[3,1,1]$) code, can correct an arbitrary single-qubit error \cite{eczoo_stab_5_1_3, eczoo_quantum_repetition}.

Singular values for both codes under different noise models is presented in \cref{tab:3Qcode:singularvalues,tab:5Qcode:singularvalues}. A common feature in both is that the bare QEC channel without any noise has only one non-trivial singular value $|\mathcal{G}|$ with multiplicity $4^k, k=1$. Under well-behaved single qubit errors that are assumed correctable by each code, an additional singular value at unity is observed.  In all other strong noise limits,  singular values and their multiplicities decay  until there is only one non-zero singular value with multiplicity one. Since this non zero value is not unity, the combined QEC and noise channel remains non-unital in the strong noise limit. 

\begin{table}
    \centering
    \begin{tabular}{p{40mm}|p{30mm}|p{30mm}}
        Singular values \newline repeated $m$ times (3 s.f.)        & No code & Noisy QEC channel (3Q) \\
        \hline
        No noise                & N/A &  $0.00, m=60$ 
                                    \newline $2.00, m=4$\\
        \hline 
        1Q Paulis               & $0.00, m=24$ 
                                \newline $0.50, m=32$ 
                                \newline $1.00, m=8$ & 
                                $0.00, m=60$ 
                                \newline $1.00, m=2$
                                \newline $2.00, m=2$\\
        
        \hline 
        3Q  Pauli group         & $0.00, m=63$ 
                                \newline $1.00, m=1$ & 
                                $0.00, m=63$ 
                                \newline $2.00, m=1$\\
                                
        % \hline 
        % Amp. Damp. ($p=0$)      & $1.00, m=64$ & 
        %                          $0.00, m=60$ 
        %                          \newline $2.00, m=4$\\
                                    
        \hline 
        Amp. Damp. ($p=1$) 
                                & $0.00, m=63$ 
                                 \newline $2.83, m=1$ &
                                 $0.00, m=63$ 
                                 \newline $2.83, m=1$
                                 \\
        \hline 
        % Amp. Damp. ($p=0.5$)    & Varied & Less varied  \\
    \end{tabular}
    \caption{Singular values with multiplicity $m$ for a channel representing a 3 qubit QEC code with a single round of stabilizer measurements composed with a noise channel.}
    \label{tab:3Qcode:singularvalues}
\end{table}

\begin{table}
    \centering
    \begin{tabular}{p{40mm}|p{30mm}|p{30mm}}
        Singular values \newline repeated $m$ times (3 s.f.) & No code & Noisy QEC channel (5Q)\\
        \hline
        No noise                    & N/A &  $0.00, m=1020$ 
                                    \newline $4.00, m=4$\\
        \hline
        
        1Q Paulis
                                & $0.00, m=405$ 
                                \newline $0.25, m=513$ 
                                \newline $0.50, m=90$ 
                                \newline $0.75, m=15$ 
                                \newline $1.00, m=1$ & 
                                $0.00, m=1020$ 
                                \newline $1.00, m=3$
                                \newline $4.00, m=1$
                                \\
        \hline 
        5Q Pauli group        & $0.00, m=1023$ 
                                \newline $1.00, m=1$ & 
                                $0.00, m=1023$ 
                                \newline $4.00, m=1$
                                \\ 
        % \hline 
        % Amp. Damp. ($p=0$)      & $1.00, m=1024$ & $0.00, m=1020$ 
        %                             \newline $4.00, m=4$\\
        \hline 
        Amp. Damp. ($p=1$) 
                                & $0.00, m=1023$ 
                                 \newline $5.66, m=1$ & 
                                 $0.00, m=1023$ 
                                 \newline $4.12, m=1$\\
        \hline 
        % Amp. Damp. ($p=0.5$)    &  & \\
    \end{tabular}
    \caption{Singular values with multiplicity $m$ for a channel representing a 5 qubit QEC code with a single round of stabilizer measurements composed with a noise channel.}
    \label{tab:5Qcode:singularvalues}
\end{table}

\section{ Random channels \label{app:ergodicqc}} 
 
We characterize singular and eigenspectral distributions of three important families of random quantum channels whose Kraus terms (\cref{eqn:krausrep2}). These families have the convenient property that randomness is defined as sampling uniformly from a relevant matrix group. Aside from the case $\kappa=1$ (unitary case), the notion of a spectral gap is well defined and we use the machinery of the previous section to provide upper bounds and estimate lower bounds for this spectral gap with concentration. 

\begin{figure}
    \centering
    \includegraphics{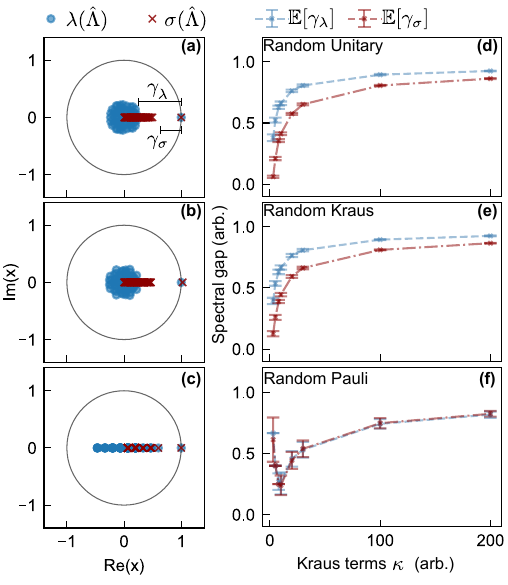}
    \caption{Spectral characteristics of random quantum channels for $d=2^3$, with rows corresponding to random unitary, random Kraus and random Pauli maps. Left panels (a)-(c): A single instance of a random quantum channel for arbitrary $\kappa=15$. Eigenvalues (blue) lie inside a complex unit disc, while singular values (red) lie between $[0,1]$ for these channels. Note that random Kraus channels are not normal nor unital; random unitary maps are not normal but are unital, and random Pauli maps are both normal and unital. A spectral gap for eigen- ($\gamma_\lambda$) and singular spectra ($\gamma_\sigma$) is defined. Right panels (d)-(f): Expected value of eigenspectral ($\gamma_\lambda$, blue) and singular spectral gaps ($\gamma_\sigma$, red) vs. increasing Kraus terms $\kappa$; in (b), the numerical quantity $\gamma_\sigma^+$ represents small approximation errors in Kraus normalization and decays to zero with increased averaging.  }
    \label{fig:raw_spectra}
\end{figure}

\begin{figure}
    \centering
    \includegraphics{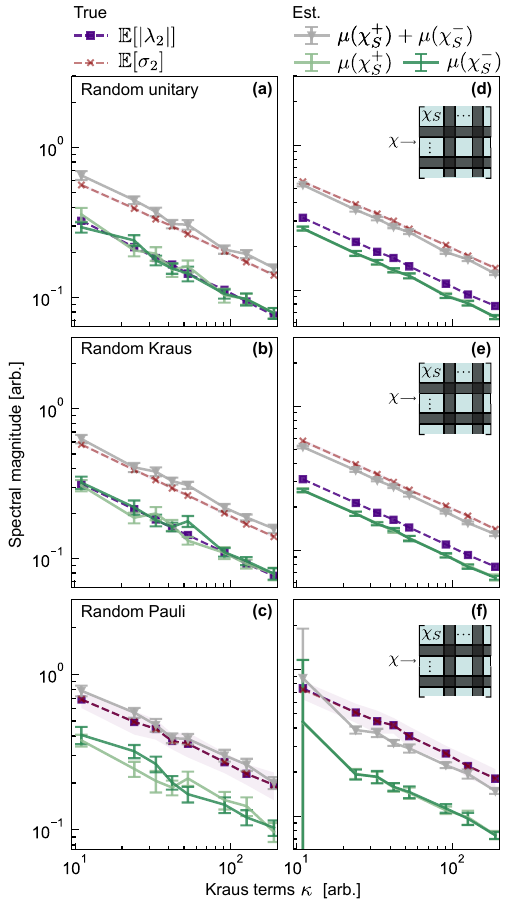}
    \caption{Upper and lower spectral bounds for true $\sigma_{2}(\hat{\Lambda})$ (red dashed) for $d=2^3$ dimensional random unitary, Kraus and Pauli quantum channels (rows) vs. Kraus terms, $\kappa$ on the $x$-axis. True  $|\lambda_2(\hat{\Lambda})|$ (indigo dashed) coincides with singular value $\sigma_{2}(\hat{\Lambda})$ only for Pauli channels. Leftmost column (a)-(c) represents the dilation $\chi(\mathbb{I} - T_p)$ of the nilpotent eigenspectrum where peripheral values of the original channel are projected out, such that Chernoff $\mu(\chi_S^+) + \mu(\chi_S^-)$ (gray solid) is an upper bound to true $\sigma_{2}(\hat{\Lambda})$ (red dashed). Rightmost column (d)-(f) additionally removes one randomly selected row/column from the dilation $\chi(\mathbb{I} - T_p)$ (setting $l_{2}=1, k=1$ in \cref{eqn:interlacing}) such that Chernoff $\mu(\chi_S^+) + \mu(\chi_S^-)$ (gray solid) interlaces to instead be a lower bound to true $\sigma_{2}(\hat{\Lambda})$ (red dashed). Empirically, the terms $\mu(\chi_S^\pm)$ (dark, light green solid) appear to lower bind the second largest absolute eigenvalue of the original channel $|\lambda_2(\hat{\Lambda})|$, even for non-normal random unitary and Kraus maps. For rightmost columns, increasing truncation ($k>1$ in \cref{eqn:interlacing}) yields looser bounds. Data represents $100$ trials of randomly generated channels for each $\kappa, d$. Individual Kraus operations  sampled from Haar unitary (a),(d), complex Ginibre (b),(e) and $3$-qubit Pauli group (c),(f) normalized to satisfy trace preservation and unique eigenvalues detected with numerical tolerance $1e^{-6}$; error bars represent one sample standard deviation.} 
    \label{fig:chernoff}
\end{figure}
\clearpage
These three groups are: the Haar ensemble of unitaries, the complex Ginibre ensemble of square matrices, and the Pauli group. The practical relevance of considering these groups stems, as examples, from the following: Haar unitary channels represent a limiting distribution for large quantum circuits with random parametrization \cite{Mele2024May}; the Ginibre ensemble represents one way to sample the space of random CPTP channels \cite{bruzda_random_2009, fischmann_induced_2012}, and Pauli channels are omnipresent in the study of noise-robust quantum computations.

To illustrate what spectral distributions can look like for these simple families, in \cref{fig:raw_spectra} we plot spectral distributions of random unitary channels, random Kraus channels \cite{kukulski_generating_2021} and random Pauli channels. An instance of the spectrum of the superoperator for a random channel is plotted in (a)-(c), with $d^2$ complex eigenvalues inside the unit circle as anticipated (blue circles). We find singular values within the real interval $[0,1]$ (red crosses), where a singular spectral norm of unity is anticipated for unital channels \cite{watrous_theory_2018}. Random unitary maps are not normal, but are unital. Meanwhile, superoperators for random Kraus maps do not appear to be normal nor unital; and random Pauli maps are both normal and unital. Furthermore, all ensembles of these maps for high $\kappa \gg 1$ appear to have a single peripheral eigenvalue at unity, with multiplicity one. (Without considering eigenvectors for $\lambda_1=1$, we cannot claim ergodicity or mixing.) 

We additionally able to define the notion of a spectral gap, $\gamma_\lambda, \gamma_\sigma$, which is the distance between the set of values whose magnitude match the spectral radius, and the next largest (absolute) value \cite{wolf_tour_2012}. For the random channels in \cref{fig:raw_spectra}, this reduces to the distance from unity. The quantities $\gamma_\lambda, \gamma_\sigma$ (blue, red) are plotted in (d)-(f) as a function of $\kappa$, with dimensions fixed at $d=2^3$. For Pauli channels (c), (f) we see normality ensures that singular values are absolute values of eigenvalues. Indeed this equivalence does not hold for non-normal channels in (d),(e). For random Kraus maps in \cref{fig:raw_spectra}(d), it is conjectured that once the $\lambda_1=1$ value is removed via the projection $\mathbb{I} - T_p$, then the remaining matrix has a spectral radius of order $1/\kappa$ \cite{kukulski_generating_2021}. However, we additionally find that spectral radius scaling of nilpotent spectrum is of order $1/\kappa$ not just for random Kraus maps, but also random unitary maps and random Pauli maps for large $\kappa$. Finally, we find that average eigenspectral gap lies below the singular spectral gap for all $\kappa$ for both random unitary and Kraus maps. These observations are empirical and we leave it as an open question to understand their theoretical origin.

We now estimate a concentrated upper bound for the second largest singular value $\sigma_{2}(\hat{\Lambda})$. For the random channels considered above, this $\sigma_{2}(\hat{\Lambda}) \geq \lambda_{2}(\hat{\Lambda})$ and directly lower bounds the spectral gap. In the high $\kappa$ regime, we may first project out the peripheral spectrum of the original channel $\hat{\Lambda}$ by using projectors $\mathcal{N}$ from the decomposed Jordan normal form in \cref{eqn:nilpotent}. Namely, \cref{th:mainresult} is applied to a complex matrix $\Gamma:= \mathcal{N}T\mathcal{N}^\dagger$ that is similar to the nilpotent spectrum of the original channel, and the corresponding dilation $\chi= \begin{bmatrix} 0 & \Gamma \\ \Gamma^\dagger & 0 \\ \end{bmatrix}$ is used in \cref{th:mainresult}. This dilation has a maximal eigenvalue which concentrates and upper binds the singular values associated with nilpotent projection of $\hat{\Lambda}$. (For normal channels, we could directly claim that this procedure upper binds the second largest absolute eigenvalue, and therefore lower binds the eigenspectral gap.)

In \cref{fig:chernoff}(a)-(c), Chernoff fluctuation term from \cref{eqn:chernoffmu}, $\mu(\chi^+) + \mu(\chi^-)$ (grey solid) is plotted for an increasing number of Kraus terms $\kappa$, to be compared with true singular value $\sigma_{2}(\hat{\Lambda})$ (red dashed) and the true second largest eigenvalue $|\lambda_{2}(\hat{\Lambda})|$ (indigo dashed). In all panels \cref{fig:chernoff} (a)-(c), we indeed find that the Chernoff fluctuation terms upper bounds the true second largest singular value $\sigma_{2}(\hat{\Lambda})$ (red dashed) of the original random quantum channel. Meanwhile the contributions of the individual Chernoff terms $\mu(\chi^\pm)$ (light, dark green) are symmetric for postive and negative parts of real spectrum of a Hermitian matrix $\chi$ in the leftmost panels.

In the rightmost panels \cref{fig:chernoff}(d)-(f), we additionally estimate a lower bound for $\sigma_{2}(\hat{\Lambda})$ (red dashed) by invoking the interlacing property via \cref{th:mainresult}. One may delete a single randomly selected row/column in the above dilation, $\chi= \begin{bmatrix} 0 & \Gamma \\ \Gamma^\dagger & 0 \\ \end{bmatrix} \to \chi_S$ for $k=1$ in \cref{eqn:interlacing} such that the maximal eigenvalue of the submatrix is a lower bound to the maximal eigenvalue of $\chi$. We numerically confirm this in \cref{fig:chernoff} (d)-(f), where the net Chernoff fluctuation (gray solid) lies just below the true value $\sigma_{2}(\hat{\Lambda})$ (red dashed), i.e. reversing the comparison in (a)-(c). There is some level of tunability in the lower bound, for example, by choosing more than one row/column to delete. Meanwhile we find that the individual terms $\mu(\chi^\pm)$ (light, dark green) in panels (d)-(f) appear to lie below the true $|\lambda_{2}(\hat{\Lambda})|$, even for non-normal superoperators of random channels. While establishing lower bounds is out of scope for \cref{th:mainresult}, theoretical results to support empirical observations in this direction may more useful in practice. 

% \bibliography{refs}
%

\end{document}